\IfFileExists{../EmanueleViolaDir.txt}{\def\EmanueleViolaDir{1}}{
\def\EmanueleViolaDir{0}}

\def\public{0}

\documentclass[letterpaper, 12pt]{article}

\usepackage{fullpage,amssymb,amsmath,amsthm,thm-restate}

\ifnum\EmanueleViolaDir=1
 \usepackage{../bibspacing}
\usepackage{../flexiblemathdisplay}
\else
  \usepackage{bibspacing}
\usepackage{flexiblemathdisplay}
\fi

\usepackage[usenames,dvipsnames]{pstricks}
\usepackage{epsfig}

\theoremstyle{definition} 
\newtheorem{theorem}{Theorem} 
\newtheorem{corollary}[theorem]{Corollary}
\newtheorem{claim}[theorem]{Claim}
\newtheorem{lemma}[theorem]{Lemma}
\newtheorem{theorem*}{Theorem}
\newtheorem{fact}[theorem]{Fact}

\newtheorem{definition}[theorem]{Definition}

\renewenvironment{proof}{\medskip \noindent {\bfseries Proof:}}{\hfill $\blacksquare$ \medskip}

\ifnum\public=1
 \newcommand{\en}[1]{}
 \newcommand{\eln}[1]{}
\else
 \newcommand{\en}[1]{\marginpar{#1}}
 \newcommand{\eln}[1]{\textbf{Note:} #1 \textbf{End note.}}
\fi

\newcommand{\poly}{\mathrm{poly}}
\renewcommand{\log}{\lg}

\newcommand{\rb}[1]{\left( #1 \right)} 

\newcommand{\zo}{\{0, 1\}}

\newcommand{\e}{\epsilon}

\newcommand{\sss}{\ensuremath{\mathrm{3SUM}}}
\newcommand{\xxx}{\ensuremath{\mathrm{3XOR}}}
\newcommand{\cxxx}{\ensuremath{\mathrm{C3XOR}}}

\newcommand{\sixs}{\ensuremath{\mathrm{6SUM}}}
\newcommand{\Sx}{\ensuremath{\mathrm{S^x}}}
\newcommand{\csss}{\ensuremath{\mathrm{C3SUM}}}
\newcommand{\fcliq}{\ensuremath{\mathrm{4Clique}}}

\begin{document}

\title{\sss , \xxx, Triangles }
\author{Zahra Jafargholi\thanks{Supported by NSF grant CCF-0845003. Email: \texttt\{zahra,viola\}@ccs.neu.edu}
   \and Emanuele Viola\footnotemark[1]}
\date{\today}

\maketitle

\begin{abstract}
We show that if one can solve $\sss$ on a set of size $n$
in time $n^{1+\e}$ then one can list $t$ triangles in a
graph with $m$ edges in time $\tilde
O(m^{1+\e}t^{1/3-\e/3})$. This is a reversal
of P{\v a}tra{\c s}cu's reduction from \sss~to listing
triangles (STOC '10).  Our result builds on and extends
works by the Paghs (PODS '06) and by Vassilevska and
Williams (FOCS '10).  We make our reductions deterministic using tools from pseudorandomness.

We then re-execute both P{\v a}tra{\c s}cu's reduction
and our reversal for the variant \xxx~of \sss~where
integer summation is replaced by bit-wise xor. As a
corollary we obtain that if \xxx~is solvable in linear
time but \sss~requires quadratic randomized time, or vice
versa, then the randomized time complexity of listing $m$
triangles in a graph with $m$ edges is $m^{4/3}$ up to a
factor $m^\alpha$ for any $\alpha > 0$.

\end{abstract}

\thispagestyle{empty}
\newpage

\setcounter{page}{1}

\section{Introduction}
The \sss~problem asks if there are three integers $a,b,c$
summing to $0$ in a given set of $n$ integers of
magnitude $\poly(n)$. This problem can be easily solved
in time $\tilde O(n^2)$. (Throughout, $\tilde O$ and
$\tilde \Omega$ hide logarithmic factors.) It seems
natural to believe that this problem also requires time
$\tilde \Omega(n^2)$, and this has been confirmed in some
restricted models.\cite{Erickson99,AilonC05} The
importance of this belief was brought to the forefront by
Gajentaan and Overmars who show that the belief implies
lower bounds for a number of problems in computational
geometry;\cite{GajentaanO95} and the list of such
reductions has grown ever since. Recently, a series of
exciting papers by Baran, Demaine, P{\v a}tra{\c s}cu,
Vassilevska, and Williams set the stage for, and
establish, reductions from \sss~to new types of problems
which are not defined in terms of
summation.\cite{BaranDP08,VassilevskaW09,PatrascuW10,Patrascu10}
In particular, P{\v a}tra{\c s}cu shows that if we can
list $m$ triangles in a graph with $m$ edges given as
adjacency list in time $m^{1.3 \bar 3 - \Omega(1)}$ then
we can solve \sss~in time $n^{2 -
\Omega(1)}$.\cite{Patrascu10}

To put this outstanding result in context we briefly
review the state-of-the-art on triangle detection and
listing algorithms. All the graphs in this paper are
undirected and simple. Given the adjacency list of a
graph with $m$ edges, Alon, Yuster, and Zwick show in
\cite{AlonYZ97} how to determine if it contains a
triangle in time $O(m^{2 \omega/(\omega+1)})$ where $\omega
< 2.376$ is the exponent of matrix multiplication. If
$\omega=2$ then the bound is $O(m^{1.3 \bar 3})$. For
listing \emph{all} triangles in a graph the best we can
hope for is time $\tilde O(m^{1.5})$, since the maximum
number of triangles in graphs with $m$ edges is
$\Theta(m^{1.5})$. There are algorithms that achieve time
$\tilde O(m^{1.5})$. (For example, we can first list the
$\le O(m \sqrt{m})$ triangles going through some node of
degree $\le \sqrt{m}$, and then the $\le O(m/\sqrt{m})^3
= O(m^{1.5})$ triangles using nodes of degree $>
\sqrt{m}$ only.) However, to list only $m$ triangles
conceivably time $\tilde O(m)$ suffices. In fact, Pagh
(personal communication 2011) points out an algorithm for
this problem achieving time $\tilde O(m^{1.5 -
\Omega(1)})$ and, assuming that the exponent of matrix
multiplication is $2$, time $\tilde O(m^{1.4})$.

\subsection{Our results} \label{s-our-results}

The main result in this paper is a ``reversal'' of P{\v
a}tra{\c s}cu's aforementioned reduction from \sss~to
listing triangles, see Corollary \ref{c-revpatrascu}
below.  In particular, we show that solving \sss~in time
$\tilde O(n^{1+\alpha})$ for some $\alpha < 1/15$ would
improve the aforementioned Pagh's triangle-listing
algorithm. (Recall the latter has complexity $\tilde
O(m^{1.4})$ assuming $\omega=2$.)  Before formally
stating this result we would like to provide some
motivation (in addition to the obvious one of filling the
landscape of reductions).

The motivation comes from the study of variants of
\sss~over other domains such as finite groups. Building on \cite{PatrascuW10}, the
paper \cite{BhattacharyyaIWX11} links such
variants to the Exponential Time Hypothesis \cite{IPZ01}
when the number of summands is ``large,'' in particular,
bigger than 3.  By contrast, we focus on the problem
which we call \xxx~and which is like \sss~except that
integer summation is replaced with bit-wise xor. So one
can think of \xxx~as asking if a given $n \times O(\lg
n)$ matrix over the field with two elements has a linear
combination of length $3$. This problem is likely less
relevant to computational geometry, but is otherwise
quite natural. Similarly to $\sss$, $\xxx$ can be solved
in time $\tilde O(n^2)$, and it seems natural to
conjecture that \xxx~requires time $\tilde \Omega(n^2)$.
But it is interesting to note that if we ask for
\emph{any} number (as opposed to $3$) of elements that
sums to $0$ the difference in domains translates in a
significant difference in complexity: SUBSET-SUM is
NP-hard, whereas SUBSET-XOR can be solved efficiently via
Gaussian elimination. On the other hand, SUBSET-XOR
remains NP-hard if the number of elements that need to
sum to $0$ is given as part of the input.\cite{Vardy97}

In light of this, it would be interesting to relate the
complexities of \sss~and \xxx. For example, it would be
interesting to show that one problem is solvable in time
$n^{2-\Omega(1)}$ if and only if the other is. Less
ambitiously, the weakest possible link would be to
exclude a scenario where, say, \sss~requires time $\tilde
\Omega(n^2)$ while \xxx~is solvable in time $\tilde
O(n)$. We are not even able to exclude this scenario, and
we raise it as an open problem.

However we manage to spin a web of reductions around
\sss, \xxx, and various problems related to triangles, a
web that extends and complements the pre-existing web.
One consequence is that the only way in which the
aforementioned scenario is possible is that listing $m$
triangles requires exactly $m^{1.3 \bar 3}$ up to
lower-order factors.

\begin{corollary} \label{c-weakestlink}
Suppose that \sss~requires randomized time $\tilde
\Omega(n^2)$ and \xxx~is solvable in time $\tilde O(n)$,
or vice versa. Then, given the adjacency list of a graph
with $m$ edges and $z$ triangles (and $O(m)$ nodes), the
randomized time complexity of listing $\min\{z, m\}$
triangles is $m^{1.3 \bar 3}$ up to a factor $m^\alpha$
for any $\alpha > 0$.
\end{corollary}

We now overview our reductions. First we build on and
extend a remarkable reduction \cite{WilliamsW10} by
Vassilevska and Williams from listing triangles to
detecting triangles. Their reduction worked on adjacency
matrixes, and a main technical contribution of this work
is an extension to adjacency lists which is needed in our
subquadratic context.

\begin{restatable}{lemma}{secondlemmarestate}
\label{lemma-second} Suppose given the adjacency list of
a graph with $m$ edges and $n=O(m)$ nodes,  one can
decide if it is triangle-free in time $m^{1+\e}$ for
$\e>0$.
Then, given the adjacency list of a graph $G$ with $m$
edges, $n=O(m)$ nodes and $z$ triangles and a positive
integer $t$ one can list $\min\{t,z\}$ triangles in $G$,
in time $\tilde O(m^{1+\e}t^{1/3-\e/3})$.
\end{restatable}

For context, Pagh shows a reduction from finding the set
of edges involved in some triangle to listing triangles,
see \cite[\S6]{Amossen11}.

Next we move to reductions between \sss~and detecting
triangles. The Paghs \cite[\S 6]{PaghP06} give an
algorithm to ``compute the join of three relations [...]
where any pair has a common attribute not shared by the
third relation.'' One component of their algorithm can be
phrased as a randomized reduction from detecting
(tripartite, directed) triangles to \sss. The same
reduction works for \xxx. Here our main technical
contribution is to exhibit a deterministic reduction,
relying on the explicit design construction by Gutfreund
and Viola \cite{GuV04}.

\begin{restatable}{lemma}{firstlemmarestate}
\label{lemma-first} Suppose that one can solve $\sss$ or
$\xxx$ on a set of size $n$ in time $\tilde O(n^{1+\e})$
for $\e\geq 0$. Then, given the adjacency list of a graph
with $m$ edges, $n=O(m)$ nodes,  one can decide if it is
triangle-free in time $\tilde O(m^{1+\e})$.
\end{restatable}

In particular, this shows that solving either \sss~or
\xxx~in time $O(m^{2 \omega/(\omega+1) - \Omega(1)})$,
where $\omega$ is the exponent of matrix multiplication,
would improve the aforementioned triangle-detection
algorithm in \cite{AlonYZ97}.

The combination of the previous two lemmas yields what we
mentioned at the beginning of \S\ref{s-our-results}: a
reversal of P{\v a}tra{\c s}cu's reduction from \sss~to
listing triangles.

\begin{restatable}{corollary}{mainthm}[Reverse P{\v a}tra{\c s}cu]
\label{c-revpatrascu} Suppose that one can solve $\sss$
or $\xxx$ on a set of size $n$ in time $n^{1+\e}$ for $\e
> 0$. Then, given the adjacency list of a graph $G$ with
$m$ edges, $n=O(m)$ nodes and $z$ triangles a positive
integer $t$, one can list $\min\{t,z\}$
triangles in $G$ in time $\tilde O(m^{1+\e}t^{1/3-\e/3})$.
\end{restatable}

Finally, we re-execute P{\v a}tra{\c s}cu's reduction for
\xxx~instead of \sss. Our execution also avoids some
technical difficulties and so is a bit simpler; it
appeared first in the manuscript \cite{Viola-xxx}.

\begin{restatable}{theorem}{patrascuxxx}
\label{t-patrascuxxx}
Suppose that given the adjacency list of a graph with $m$
edges and $z$ triangles (and $O(m)$ nodes) one can list
$\min\{z, m\}$ triangles in time $m^{1.3\bar3 - \e}$ for
a constant $\e > 0$. Then one can solve $\xxx$ on a set
of size $n$ in time $n^{2-\e'}$ with error $1\%$ for a
constant $\e' > 0$.
\end{restatable}

\subsection{Techniques}

We now explain how we reduce listing $t$ triangles in a
graph to detecting triangles. First we recall the
strategy \cite{WilliamsW10} by Vassilevska and Williams
that works in the setting of adjacency matrixes, as
opposed to lists.

Without loss of generality, we work with a tripartite
graph with $n$ nodes per part. Their recursive algorithm
proceeds as follows. First, divide each part in two
halves of $n/2$ nodes, then recurse on each of the $8$
subgraphs consisting of triples of halves. Note edges are
duplicated, but triangles are not. One uses the
triangle-detection algorithm to avoid recursing on
subproblems that do not contain triangles. The most
important feature is that one keeps track, throughout the
execution of the algorithm, of how many subproblems have
been produced, and if this number reaches $t$ one stops
introducing new subproblems.

We next explain how to extend this idea to adjacency
lists. At a high level we use the same recursive approach
based on partitions and keeping track of the total number
of subproblems. However in our setting partitioning is
more difficult, and we resort to \emph{random
partitioning}. Each part of the graph at hand is
partitioned in $2$ subsets by flipping an unbiased coin
for each node. If we start with a graph with $m$ edges,
each of the $8$ subgraphs (induced by the $8$ triples of
subsets) expects $m/4$ edges. We would like to guarantee
this result with high probability simultaneously for each
of the $8$ subgraphs. For this goal we would like to show
that the size of each of the $8$ subgraphs is
concentrated around its expectation. Specifically, fix a
subgraph and let $X_e$ be the indicator random variable
for edge $e$ being in the subgraph. We would like to show
\[ \Pr\left[\sum_e X_e > m/4 + \gamma m\right] < 1/8 \hspace{1cm} (\star) \]
for some small $\gamma$. By a union bound we can then
argue about all the $8$ subgraphs simultaneously.

Assuming $(\star)$ we conclude the proof similarly to
\cite{WilliamsW10} as follows, setting for simplicity
$\gamma = 0$. Each recursive step reduces the problem
size by a factor $4$. Let $s_i$ be the number of
subproblems at level $i$ of the recursion. The running
time of the algorithm is of the order of
\[ \sum_{i \le \lg_4 m} s_i T(m/4^i) < \sum s_i
m^{1+\e}/4^i,\] where $T(x) = x^{1+\e}$ is the time of
the triangle detection algorithm. Since we recurse on
$\le 8$ subproblems we have $s_i \le 8^i$; since we make
sure to never have more than $t$ subproblems we have $s_i
\le t$. Picking a breaking-point $\ell$ we can write the
order of the time as
\[ m^{1+\e} \rb{\sum_{i \le \ell } 8^i /4^i + \sum_{i >
\ell}^{\lg_4 m} t /4^\ell} = m^{1+\e} \tilde O(2^\ell +
t/4^\ell)\] which is minimized to $\tilde O(m^{1+\e} t^{1/3})$ for $\ell = \lg t^{1/3}$.

We now discuss how we guarantee $(\star)$. The obstacle
is that the variables $X_e$ are not even pairwise
independent; consider for example two edges sharing a
node. We overcome this obstacle by introducing a first
stage in the algorithm in which we list all triangles
involving at least one node of high degree ($>
\delta m$), which only costs time $\tilde O(m)$. We then
remove these high-degree nodes. What we have gained is
that now most pairs of variables $X_e, X_{e'}$ are
pairwise independent. This lets us carry through an
argument like Chebychev's inequality and in turn argue
about concentration around the expectation $m/4$.

To obtain a deterministic reduction we choose the
partition from an almost 4-wise independent sample space
\cite{NaN93,AGHP92}.

\paragraph{Organization.}
In \S\ref{ExistTriangle-reducing-to-sss} we reduce
detecting triangles to \sss~and \xxx. In
\S\ref{ListingTriangle-reducing-to-ExistTriangle} we give
the reduction from listing to detecting triangles in a
graph. In the appendix,
\S\ref{Reducing-xxx-to-ListingTriangles} through
\S\ref{s-reducing-to-triangles}, we give the reduction
from \xxx~to listing triangles.  In another section of
the appendix, \S\ref{s-4-clique-6-sum}, we show how to
reduce 4-clique to \sixs~over the the group
$\mathbb{Z}_3^t$ -- thus hinting at a richer web of
reductions.

Note that Corollary \ref{c-weakestlink} follows
immediately from the combination of: P{\v a}tra{\c s}cu's
reduction from \sss~to listing triangles
\cite{Patrascu10}, our reversal (Corollary
\ref{c-revpatrascu}), and our re-execution for
\xxx~(Theorem \ref{t-patrascuxxx}).

\section{Reducing detecting triangles to \sss~and \xxx }
\label{ExistTriangle-reducing-to-sss}

In this section we prove Lemma \ref{lemma-first},
restated next.

\firstlemmarestate*

Recall that all graphs in this paper are undirected (and
simple). Still, we use ordered-pair notation for edges. A
triangle is a set of $3$ distinct edges where each node
appears twice, such as $(a,b), (c,b), (a,c)$.

The deterministic reduction relies on combinatorial
designs, family of $m$ subsets of a small universe with
small pairwise intersections. Specifically we need the
size of the sets to be linear in the universe size, and
the bound on the intersection a constant fraction of the
set size. Such parameters were achieved in \cite{NiW94}
but with a construction running in time exponential in
$m$. We use the different construction computable in time
$\tilde O(m)$ by Gutfreund and Viola \cite{GuV04}.

\begin{lemma}{\cite{GuV04}} \label{lemma-design}
For every constant $c>1$ and large enough $m$ there is a
family of $m$ sets $S_i, i = 1,\ldots,m$ such that:

1) $|S_i| = c^2 \log m$,

2) $|S_i \cap S_j| \le 2c \log m$ for $i \ne
j$,

3) $S_i \subseteq [50 \cdot c^3 \log m]$,

4) the family may be constructed in time $\tilde O(m)$.
\end{lemma}

Note that by increasing $c$ in Lemma \ref{lemma-design}
we can have the bound on the intersection size be an
arbitrarily small constant fraction of the set size.

\begin{proof} 
We show how to reduce detecting directed triangles to
\sss. The same approach reduces detecting un-directed
triangles to \xxx~(except that the numbers below would be
considered in base 2 instead of 10). To reduce detecting
un-directed triangles to \sss, we can simply make our
graphs directed by repeating each edge with direction
swapped.

We first review the randomized reduction, then we make it
deterministic. Given the adjacency list of graph
$G=(V,E)$, assign an $\ell$-bit number, uniformly and
independently to each node in the graph $G$, $\ell$ to be
determined, i.e. $\forall a\in V,\quad X_a  \in_U
\{0,1\}^\ell.$ For each edge $e=(a,b)$ let $Y_{(a,b)}:=
(X_a - X_b)$. Return the output of $\sss$ on the set
$Y:=\{Y_{(a,b)}| (a,b) \in E\}$. If there is a triangle
there are always 3 elements in $Y$ summing to $0$.
Otherwise, by a union bound the probability that there
are such 3 elements is $\le 1/2$ for $\ell = 3 \log m$.

To make the reduction deterministic, consider the family
$S_i$ of $O(m)$ sets from Lemma \ref{lemma-design}, with
intersection size less than 1/5 of the set size. Assign
to node $a$ the number $x_a$ whose decimal representation
has 1 in the digits that belong to $S_a$, and 0
otherwise.

As before, we need to show that if there are 3 numbers
$(x_a - x_{a'}), (x_b - x_{b'})$ and $(x_c - x_{c'})$ in $Y$
that sum to $0$ then there is a triangle in the graph.
Since the graph has no self loops, note that the
existence of a triangle is implied by the fact that in
the expression $(x_a - x_{a'}) + (x_b - x_{b'}) + (x_c -
x_{c'})$ each of $x_a, x_b, x_c$ appears exactly once
with each of the two signs. We will show the latter.
First, we claim that each of $x_a, x_b, x_c$ appears the
same number of times with each sign. Indeed, otherwise
write the equation $x_a + x_b + x_c = x_{a'} + x_{b'} +
x_{c'}$ and simplify equal terms. We are left with a
number on one side of the equation that is non-zero in a
set of digits that cannot be covered by the other 5, by
the properties of the design. Hence the equation cannot
hold. Note that when performing the sums in this equation
there is no carry among decimal digits. Finally, we claim
that no number can appear twice with the same sign. For
else it is easy to see that there would be a self loop.
\end{proof}

\section{Reducing listing to detecting triangles}
\label{ListingTriangle-reducing-to-ExistTriangle}

In this section we prove Lemma \ref{lemma-second},
restated next.

\secondlemmarestate*

\begin{proof}
Let $A$ denote the triangle-detection algorithm. The
triangle-listing algorithm is called $B$ and has three
stages. In stage one, we list triangles in $G$ involving
at least one high degree node. In this part we do not use
algorithm $A$.
 In the second stage, we create a new graph
$G'$ which is tripartite, and has the property that for
each triangle in $G$ there uniquely correspond $6$ triangles in $G'$.
 In the final stage we run a recursive algorithm on
$G'$ and list $\min\{6t,6z\}$ triangles in $G'$ which
would correspond to $\min\{t,z\}$ triangles in $G$. This
recursive algorithm will make use of algorithm $A$.

\paragraph{Stage One.}  We consider a node to be high degree
if its degree is $> \delta m$, for a parameter $\delta$
to be set later. We can list triangles involving a high
degree node, if any exists, in time $\tilde O(m/\delta)$.
To see this, note that we can sort the adjacency list and
also make a list of high degree nodes in time $\tilde
O(m)$. Also note that the number of nodes with high
degree is $O(1/\delta)$, because the sum of all degrees
is $2m$. For any high degree node $h$, for each edge
$(a,b)\in E$ we search for two edges $(a,h)$ and $(b,h)$
in the adjacency list. Since the adjacency list is
sorted, the search for each edge will take $\tilde O(
\log m)$ and for each high degree nodes we do this search
$2m$ times so the running time of Stage One is
$T_{B_1}(m)=\tilde O(m/\delta)$. Obviously at any point
of this process, if the number of listed triangles
reaches $t$ we stop. If not, we remove the high degree
nodes from $G$ and move to the next stage.

\paragraph{Stage Two.} We  convert $G$ into a tripartite graph
$G'=(V',E')$ where $V':=I_1\cup I_2\cup I_3$ and each
part of $V'$ is a copy of $V$. For each edge $(a,b)$ in
$E$ place in $E'$ edge $(a_i,b_j)$ for any $i, j \in
\{1,2,3\}, i\ne j$.

A triangle in $G$ yields $6$ in $G'$ by any choice of the
indices $i$ and $j$. A triangle in $G'$ yields one in $G$
by removing the indices, using that the graph is simple.
This stage takes time $T_{B_2}(m)=\tilde O(m)$. In the
next stage we look for $t'=6t$ triangles in $G'$. Note
that $|E'|=6|E|$.

\paragraph{Stage Three.}
We partition each of $I_1, I_2$ and $I_3$ of $V'$
randomly into two subsets, in a way specified below. Now
we have 8 subgraphs, where each subgraph is obtained by
choosing one subsets from each of $I_1, I_2$ and $I_3$.
For each of the subgraphs, we use $A$ to check if the
subgraph contains a triangle. If it does, we recurse on
the subgraph. We recurse till the number of edges in the
subgraph is smaller than a constant $C$, at which point
by brute force in time $\tilde O(C^3)$ we return all the
triangles in the subgraph. Note that each triangle
reported is unique since it only appears in one
subproblem. We only need to list $t'$ triangles in the
graph, so when the number of subproblems that are
detected to have at least one triangle reaches $t'$, we
do not need to introduce more.

To bound the running time, we need to bound the size of
the input for each subproblem. If the random partition
above is selected by deciding uniformly and independently
for each node which subset it would be in, the expected
number of edges in each subgraph is $m/4$. We introduce
another parameter $\gamma$ and consider the probability
that all the $8$ subproblems are of size smaller than
$m/4+m\gamma$. We call these subproblems roughly
balanced.

To make the reduction deterministic we choose the
partition by an almost $4$-wise independent space
\cite{NaN93,AGHP92}.

\begin{lemma}[\cite{NaN93,AGHP92}] \label{l-almost-indep}
There is an algorithm that maps a seed of $O(\log \log n
+ k + \log 1/\alpha)$ bits into $n$ bits in time $\tilde
O(n)$ such that the induced distribution on any $k$ bits
is $\alpha$-close to uniform in statistical distance.
\end{lemma}

\begin{restatable}{claim}{claimrestate}
\label{claim-1} Let $0<\gamma <1/4$. There are $\delta$
and $\alpha$ such that for all sufficiently large $m$, if
we partitioning each of $I_1, I_2$ and $I_3$ into two
subsets using an $\alpha$-almost $4$-wise independent
generator, with probability $> 0$ all the $8$ subgraphs
induced by triples of subsets have less than
$m(1/4+\gamma)$ edges.
\end{restatable}

We later give the proof of this claim. To make sure that
all the subproblems generated during the execution of the
entire algorithm are roughly balanced, each time we
partition we enumerate all seeds for the almost $4$-wise
independent generator, and pick the first yielding the
conclusion of Claim \ref{claim-1}. This only costs
$\tilde O(m)$ time.

To analyze the running time of Stage Three, let $s_i$
denote the number of subproblems at level $i$ of the
recursion. At the $i$th level, we run algorithm $A$ $s_i$
times on an input of size $\leq
6m\left(1/4+\gamma\right)^i$, so the running time of the
recursive algorithm at level $i$ is $\tilde O\left(s_i
\cdot T_A \left( 6 m\left( 1/4+\gamma \right)^i \right)
\right)$, where $T_A(m)= m^{1+\e}$ is the running time of
algorithm $A$.

Note that $s_i\leq 8^i$ by definition and $s_i\leq t'$
because at any level we keep at most $t'$ subproblems. Pick $\ell := \lg t^{1/3}$. The running time of this
stage is

\begin{align*}
T_{B_3}(6m,6t) & = \tilde O \left( \sum_{i=0}^ {\ell-1} 8^i T_A \left( 6m\left( 1/4+\gamma \right)^i \right)+\sum_{i=\ell}^ {\log 6m} 6t\cdot T_A\left(6m\left(1/4+\gamma \right)^i \right)+6tC^3\right)\\
& =\tilde O\left( \sum_{i=0}^ {\ell-1} 8^i \left( m\left( 1/4+\gamma \right)^i \right)^{1+\e}+\sum_{i=\ell}^ {\log 6m} t\cdot \left(m\left(1/4+\gamma\right)^i\right)^{1+\e}\right).\\
\end{align*}

The asymptotic growth of each sum is dominated by their value for $i=\ell$. 

\begin{align*}
T_{B_3}(6m,6t)& =\tilde O\left(8^\ell \cdot  m^{1+\e} \cdot (1/4+\gamma)^{\ell(1+\e)} + t \cdot m^{1+\e} \cdot (1/4+\gamma)^{\ell(1+\e)} \right)\\
& =\tilde O\left(m^{1+\e} \cdot t^{(1/3) \lg_2 (1/4+\gamma) (1+\e)+1}\right).\\
\end{align*}
Let  $\lg_2 (1/4+\gamma) = -2 + \beta$.
\begin{align*}
T_{B_3}(6m,6t) & =\tilde O\left(m^{1+\e} \cdot t^{(1/3)(-2+\beta)(1+\e)+1}\right)\\
 & =\tilde O\left(m^{1+\e}t^ {(1/3) (1- 2\e + \beta \e +\beta) } \right).\\
\end{align*}

For small enough $\gamma$,  $1-2\e+\beta\e + \beta< 1-\e$, hence:
\begin{align*}
T_{B_3}(6m,6t) & =\tilde O\left(m^{1+\e}t^ {(1/3)(1-\e)} \right).\\
\end{align*}


Finally the running time of algorithm $B$ is
\[T_B(m,t)= T_{B_1}\left( m\right)+T_{B_2}\left(m\right)+T_{B_3}\left(6m,6t\right)=\tilde O\left(m^{1+\e}t^{1/3-\e/3)}\right).\]
\end{proof}

\claimrestate*
\begin{proof}[of claim \ref{claim-1}] Let us fix one of the subgraphs and call it $S$ and define the following random variables,
\[
  \forall 0\leq i \leq m, \quad X_i= \left\{
  \begin{array}{l l}
    1 & \quad \text{if } e_i\in S,\\
   0 & \quad \text{if } e_i \notin S.\\
  \end{array} \right.
\]
We have $|E[X_i] - 1/4| \le \alpha $ and $|E[\sum_i
X_i]-m/4| \le \alpha m$. To prove the claim, we show that
the probability that $S$ has more than $m(1/4+\gamma)$
edges is less than $1/16$; and by a union bound we
conclude. In other words we need to show:
\[ P_S := \Pr\left[\sum_i X_i-m/4 \geq m\gamma\right] \leq 1/16.\]
By a Markov bound we have,
\begin{align*}
P_S\leq \Pr\left[\left(\sum_i X_i-m/4\right)^2 \geq \left(m\gamma\right)^2\right]& \leq E\left[\left(\sum_i X_i-m/4\right)^2\right]/ \left(m\gamma\right)^2.
\end{align*}
Later we bound $E\left[\left(\sum_i
X_i-m/4\right)^2\right]=O((\alpha+\delta)m^2)$ from which
the claim follows.

Now we get the bound for $E\left[\left(\sum_i
X_i-m/4\right)^2\right]$.
\begin{align*}
E\left[\left(\sum_i X_i-m/4\right)^2\right]&=E\left[\left(\sum_i X_i\right)^2+(m/4)^2-\left(m \sum_i X_i\right)/2\right]\\
&\le E\left[\sum_{i \neq j} X_iX_j\right]+E\left[\sum_i X_i^2\right] + \frac{m^2}{16} - \frac{m}2 m\rb{\frac14 - \alpha}\\
&\le E\left[\sum_{i \neq j} X_iX_j\right]+\frac{m}{4} + \alpha m -\frac{m^2}{16} + m^2 \alpha/2 \\
& = E\left[\sum_{i \neq j} X_iX_j\right]+ O(\alpha m^2) -\frac{m^2}{16}.
\end{align*}
$E\left[ X_iX_j\right]$ is the probability that two edges
$e_i$ and $e_j$ are both in $S$. If our distribution were
uniform the probability would be $1/16$ for the pairs of
edges that do not share a node, and $1/8$ for the pairs
of edges that do share a node. Let $\rho$ be the number
of unordered pairs of edges that share a node. We have:
\[
E\left[\sum_{i \neq j} X_iX_j\right] = \sum_{i \neq
j}E\left[X_iX_j\right] \le 2 \rho(1/8 + \alpha) +
2\rb{\binom{m}{2}-\rho}(1/16+\alpha) \\ \le m^2/16 +
\rho/8 + 4\rho \alpha + \alpha m^2 \le m^2/16 +
\rho/8 + O(\alpha m^2).
\]
Note
\[\rho = \sum_{a\in V} \binom{\mathrm{degree}(a)}{2} \le
\sum_{a\in V} \mathrm{degree}(a)^2 /2 \le \delta m
\sum_{a\in V} \mathrm{degree}(a) /2 \le \delta m^2,\]
since after stage one of the algorithm there are no nodes
with degree more than $\delta m$.

Hence we obtain
\[
E\left[\left(\sum_i X_i-m/4\right)^2\right] \leq
\frac{m}4 + O(\alpha + \delta)m^2 = O(\alpha + \delta)m^2,\] as desired.
\end{proof}

\paragraph{Acknowledgments.}
We are very grateful to Rasmus Pagh and Virginia
Vassilevska Williams for answering many questions on
finding triangles in graphs. Rasmus also pointed us to
\cite{PaghP06,Amossen11}. We also thank Siyao Guo for
pointing out that a step in a previous proof of Lemma
\ref{lemma-listing-3xor-hard} was useless, and Ryan
Williams for stimulating discussions.

\bibliographystyle{alpha}
{\small \ifnum\EmanueleViolaDir=1
  \bibliography{../OmniBib}
\else
  \bibliography{OmniBib}
\fi
}

\appendix

\section{Reducing $\xxx$ to listing
triangles}\label{Reducing-xxx-to-ListingTriangles}

In this section we prove theorem \ref{t-patrascuxxx}.
\patrascuxxx*

The proof of Theorem \ref{t-patrascuxxx} follows the one
in \cite{Patrascu10} for $\sss$, which builds on results
in \cite{BaranDP08}. However the proof of Theorem
\ref{t-patrascuxxx} is a bit simpler. This is because it
avoids some steps in \cite{BaranDP08,Patrascu10} which
are mysterious to us. And because in our context we have
at our disposal hash functions that are \emph{linear},
while over the integers one has to work with
``almost-linearity,'' cf.~\cite{BaranDP08,Patrascu10}.

The remainder of this section is organized as follows. In
\S\ref{s-hashing} we cover some preliminaries and prove a
hashing lemma by \cite{BaranDP08} that will be used
later.\footnote{In \cite{BaranDP08,Patrascu10} they
appear to use this lemma with a hash function that is not
known to satisfy the hypothesis of the lemma. However
probably one can use instead similar hash functions such
as one in \cite{Dietzfelbinger96} that does satisfy the
hypothesis. We thank Martin Dietzfelbinger for a
discussion on hash functions.} The proof of the reduction
in Theorem \ref{t-patrascuxxx} is broken up in two
stages. First, in \S\ref{s-cxxx} we reduce $\xxx$ to the
problem $\cxxx$ which is a ``convolution'' version of
$\xxx$. Then in \S\ref{s-reducing-to-triangles} we reduce
$\cxxx$ to listing triangles.

\subsection{Hashing and preliminaries} \label{s-hashing}

We define next the standard hash function we will use.

\begin{definition}\label{d-hash} For input length $\ell$ and output length $r$, the hash function $h$ uses $r$ $\ell$-bit keys $\overline a := (a^1, \ldots, a^r)$ and is defined as $h_{\bar a}(x) := (\langle a^1, x \rangle, \ldots, \langle a^r, x \rangle)$, where $\langle .,. \rangle$ denotes inner product modulo $2$.
\end{definition}

We note that this hash function is linear: $h_{\bar a}(x) + h_{\bar a}(y) = h_{\bar a}(x+y)$ for any $x \ne y \in \zo^\ell$, where addition is bit-wise xor. Also, $h_{\bar a}(0) = 0$ for any $\bar a$, and $\Pr_{\bar a}[h_{\bar a}(x) = h_{\bar a}(y)] \le 2^{-r}$ for any $x \ne y$.

\medskip

Before discussing the reductions, we make some remarks on the problem \xxx. First, for simplicity we are going to assume that the input vectors are unique. It is easy to deal separately with solutions involving repeated vectors. Next we argue that for our purposes the length $\ell$ of the vectors in instances of \xxx~can be assumed to be $(2-o(1)) \lg n \le \ell \le 3 \lg n$. Indeed, if $\ell \le (2- \Omega(1))\lg n$ one can use the fast Walsh-Hadamard transform to solve \xxx~efficiently, just like one can use fast Fourier transform for \sss, cf.~\cite[Exercise 30.1-7]{CLRS01}. For \xxx~one gets a running time of $2^\ell \ell^{O(1)} + \tilde O(n \ell)$, where the first term comes from the fast algorithms to compute the transform, see e.g.~\cite[\S 2.1]{MaslenR97}. (The second term accounts for preprocessing the input.) When $\ell \le (2- \Omega(1))\lg n$, the running time is $n^{2-\Omega(1)}$, i.e., subquadratic.

Also, the length can be reduced to $3 \lg n$ via hashing. Specifically, an instance $v_1,\ldots,v_n \in \zo^\ell$ of \xxx~is reduced to $h(v_1),\ldots,h(v_n)$ where $h = h_a$ is the hash function with range of $r = 3 \lg n$ bits for a randomly chosen $\bar a$. Correctness follows because on the one hand if $v_i + v_j + v_k = 0$ then $h(v_i) + h(v_j) + h(v_k) = h(v_i + v_j + v_k) = h(0) = 0$ by linearity of $h$ and the fact that $h(0) = 0$ always; on the other hand if $v_i + v_j + v_k \ne 0$ then $\Pr[h(v_i + v_j + v_k) = 0] = 1/2^r$ since $h$ maps uniformly in $\zo^r$ any non-zero input. Hence by a union bound over all $\le \binom{n}{3}$ choices for vectors such that $v_i + v_j + v_k \ne 0$, the probability of a false positive is $\binom{n}{3}/n^3 < 1/6$.

\medskip

For the proof we need to bound the number of elements $x$ whose buckets
$B_h(x) := \{y \in S : h(x) = h(y)\}$
have unusually large load. If our hash function was $3$-wise independent the desired bound would follow from Chebyshev's inequality. But our hash function is only pairwise independent, and we do not see a better way than using a hashing lemma from \cite{BaranDP08} that in fact relies on a weaker property, cf.~the discussion in \cite{BaranDP08}.

When hashing $n$ elements to $[R] = \{1,2,\ldots,R\}$, the expected load of each bucket is $n/R$. The lemma guarantees that the expected number of elements hashing to buckets with a load $\ge 2n/R + k$ is $\le n/k$.

\begin{lemma}[\cite{BaranDP08}] \label{lemma-BaranDP}
Let $h$ be a random function $h : U \to [R]$ such that for any $x \ne y$, $\Pr_h[h(x) = h(y)] \le 1/R$. Let $S$ be a set of $n$ elements, and denote $B_h(x) = \{y \in S : h(x) = h(y)\}$. We have $$\Pr_{h,x}[|B_h(x)| \ge 2n/R + k] \le 1/k.$$
In particular, the expected number of elements from $S$ with $|B_h(x)| \ge 2n/R + k$ is $\le n/k$.
\end{lemma}

The proof of the lemma uses the following fact, whose proof is an easy application of the Cauchy-Schwarz inequality.

\begin{fact}\label{fact-collision} Let $f : D \to [R]$ be a function. Pick $x,y$ independently and uniformly in $D$. Then $\Pr_{x,y}[f(x) = f(y)] \ge 1/R$.
\end{fact}

\begin{proof}[of Lemma \ref{lemma-BaranDP}]
Pick $x,y$ uniformly and independently in $S$ ($x =y$ possible). For given $h$, let
\begin{align*}
p_h & := \Pr_x[|B(x)| \ge 2n/R + k],\\
q_h & := \Pr_{x,y}[h(x) = h(y)].
\end{align*}
Note we aim to bound $E[p_h] \le 1/k$, while by assumption
\begin{equation} \label{eHashEgh}
E[q_h] = \Pr_{h,x,y}[h(x) = h(y)] \le 1/R + 1/n.
\end{equation}

Now let $L_h := \{x : |B_h(x)| < 2n/R + k\}$, and note $|L_h| = (1-p_h)n$.
Let us write
$$q_h = \Pr_{x,y}[h(x) = h(y) | x \in L_h] \Pr[x \in L_h] +
\Pr_{x,y}[h(x) = h(y) | x \not \in L_h] \Pr[x \not \in L_h].$$

The latter summand is $\ge ((2n/R + k)/n) p_h = (2/R + k/n)p_h$.

For the first summand, note
$$\Pr_{x,y}[h(x) = h(y) | x \in L_h] \Pr[x \in L_h]
= \Pr_{x,y}[h(x) = h(y) | x \wedge y \in L_h] \Pr[x \wedge y \in L_h]$$
because if $y \not \in L_h$ then there cannot be a collision with $x \in L_h$. The term $\Pr_{x,y}[h(x) = h(y) | x \wedge y \in L_h]$ is $\ge 1/R$ by Fact \ref{fact-collision}. The term $\Pr[x \wedge y \in L_h]$ is $(1-p_h)^2 \ge 1-2p_h$.

Overall,
$$q_h \ge \frac{1}R (1-2p_h) + (2/R + k/n) p_h = p_h k/n + 1/R.$$

Taking expectations and recalling \eqref{eHashEgh},
$$E[p_h] k/n + 1/R \le 1/R + 1/n,$$
as desired.
\end{proof}

\subsection{Convolution \xxx} \label{s-cxxx}

Define the problem \emph{convolution \xxx}, denoted \cxxx, as: Given array $A$ of $n$ strings of $O(\lg n)$ bits, determine if $\exists i,j \le n : A[i] + A[j] = A[i+j]$. Again, sum is bit-wise xor.

\begin{lemma}\label{lemma-xxx-cxxx} If \cxxx~can be solved with error $1\%$ in time $t \le n^{2-\Omega(1)}$, then so can $\xxx$.
\end{lemma}

\paragraph{Intuition.}
We are given an instance of 3XOR consisting of a set $S$ of $n$ vectors. Suppose for any $x \in S$ we define $A[x] := x$, and untouched elements of $A[x]$ are set randomly so as to never participate in a solution.

Now if $x+y=z$ then $A[x]+A[y] = A[z] = A[x+y]$. Using again $x+y=z$ we get $A[x]+A[y] = A[x+y]$. Hence this solution will be found in C3XOR. Conversely a solution to C3XOR corresponds to a 3XOR solution, since $A$ is filled with elements with $S$ (and random otherwise).

This reduction is correct. But it is too slow because the size of $A$ may be too large.

In our second attempt we try to do as above, but make sure the vector $A$ is small. Suppose we had a hash function $h : S \to [n]$ that was both 1-1 and linear.

Then we could let again $A[h(x)] := x$.

If $x+y=z$ then $A[h(x)] + A[h(y)] = A[h(z)]$ by definition. And using again $x+y=z$ and linearity, we get $h(x+y) = h(x)+h(y) = h(z)$, and so we get
$A[h(x)] + A[h(y)] = A[h(x) + h(y)]$ as desired.

But the problem is that there is no such hash function. (Using linear algebra one sees that there is no hash function that shrinks and is both linear and 1-1.)

The solution is to implement the hash-function based solution, and handle the few collisions separately.

\begin{proof}
Use the hash function $h$ from Definition \ref{d-hash} mapping input elements of $\ell = O(\lg n)$ bits to $r := (1-\alpha) \lg n$ bits, for a constant $\alpha$ to be determined. So the range has size $R = 2^r = n^{1-\alpha}$. By Lemma \ref{lemma-BaranDP}, the expected number of elements falling into buckets with more than $t := 3n/R$ elements is $\le R$. For each of these elements, we can easily determine in time $\tilde O(n)$ if it participates to a solution. The time for this part is $\tilde O(R n)$ with high probability, by a Markov bound.

It remains to look for solutions $x+y+z=0$ where the three elements all are
hashed to not-overloaded buckets. For each $i,j,k \in [t]$, we look for a
solution where $x,y,z$ are respectively at positions $i,j,k$ of their buckets.
Specifically, fill an array $A$ of size $O(R)$ as follows: put the $i$th ($j$th,
$k$th) element $x$ of bucket $h(x)$ at position $A[h(x)01]$ ($A[h(x)10],
A[h(x)11]$), where $h(x)01$ denotes the concatenation of the bit-strings $h(x)$
and $01$. The untouched elements of $A$ are set to a value large enough that it
can be easily shown they cannot participate in a solution. Run the algorithm for
C3XOR on $A$.

If there is a solution $x+y+z=0$, suppose $x,y,z$ are the $i$th ($j$th, $k$th) elements of their buckets. Then for that choice of $i,j,k$ we have
$A[h(x)01] = x, A[h(y)10] = y, A[h(z)11]  = z$, and so
$A[h(x)01] + A[h(y)10] = A[h(z)11]$.
By linearity of $h$, and the choice of the vectors $01,10,11$, we get $h(z)11 = h(x)01 + h(y)10$. So this solution will be found.

Conversely, any solution found will be a valid solution for $3XOR$, by construction of $A$.

The time for this part is as follows. We run over $t^3 = O(n^3/R^3)$ choices for the indices. For each choice we run the $\cxxx$ algorithm on an array of size $O(R)$. If the time for the latter is $R^{2-\e}$, we can pick $R = n^{1-\alpha}$ for a small enough $\alpha$ so that the time is $\tilde O(n^{3\alpha} n^{(2-\e)(1-\alpha)}) = n^{2-\e'}$. (Here we first amplify the error of the $\cxxx$ algorithm to $1/n^3$ by running it $O(\lg n)$ times and taking majority.)

The first part only takes time $O(Rn) = O(n^{2-\alpha})$, so overall the time is $n^{2-\e''}$.
\end{proof}

It is worth mentioning that although Lemma \ref{lemma-xxx-cxxx} shows that $\cxxx$ is at least as hard as $\xxx$, we can easily prove that is not any harder than $\xxx$ either.
\begin{lemma}\label{lemma-cxxx-xxx}
If \xxx~can be solved in time $t \le n^{2-\Omega(1)}$, then so can $\cxxx$.
\end{lemma}

\begin{proof}
Let array $A$ of $n$ elements be the input to $\cxxx$, create set $S:=\{\ A[i]i\  | \ \forall i \in [n]\}$ where $ A[i]i$ denotes the concatenation of bit-strings of $i$ and $A[i]$. Run $\xxx$ on the set $S$. It is easy to see that if \[ \exists\  a,b,c \in S \text{ such that } a+b+c=0 \iff \exists \  i,j \in [n]   \text{ such that } A[i]+A[j]= A[i+j] .\]
\end{proof}

A similar method can be applied to reduce $\csss$ to $\sss$. The only difference is in creating the elements of $S$, $S:= \{\ A[i]0i\ |\ \forall i \in [n]\}$. The 0-bit in between $A[i]$ and $i$ is to ensure that the (possible) final carry of the sum of the indices is not added to the sum of the elements of $A$.

\subsection{Reducing $\cxxx$ to listing triangles} \label{s-reducing-to-triangles}

\begin{lemma} \label{lemma-listing-3xor-hard}
Suppose that given the adjacency list of a graph with $m$ edges and $z$ triangles (and $O(m)$ nodes) one can list $\min\{z, m\}$ triangles in time $m^{1.3\bar3 - \e}$ for a constant $\e > 0$. Then one can solve $\cxxx$ on a set of size $n$ with error $1\%$ in time $n^{2-\e'}$ for a constant $\e' > 0$.
\end{lemma}

In fact, the hard graph instances will have $n = m^{1-\Omega(1)}$ nodes.

\begin{proof}
We are given an array $A$ and want to know if $\exists a,b \le n : A[a] + A[b] = A[a+b]$. It is convenient to work with the equivalent question of the existence of $a,b$ such that $A[a+b_h] + A[a+b_\ell] = A[b]$, where $b_h, b_\ell$ are each half the $\lg n$ bits of $b$.

We use again the linear hash function $h$. To prove Lemma \ref{lemma-xxx-cxxx} we hashed to $R = n^{1-\e}$ elements. Now we pick $R := \sqrt{n}$. By the paragraph after Definition \ref{d-hash}, among the $\le n^2$ pairs $a,b$ that do not constitute a solution (i.e., $A[a+b_h] + A[a+b_\ell] \ne A[b])$, we expect $\le n^2/R$ of them to satisfy
$$h(A[a+b_h]) + h(A[a+b_\ell]) = h(A[b]) \hspace{2cm}  (\star).$$

By a Markov argument, with constant probability there are $\le 2 n^2/R = 2 n^{1.5}$ pairs $a,b$ that do not constitute a solution but satisfy $(\star)$. The reduction works in that case. (One can amplify the success probability by repetition.)

We set up a graph with $m := 3 n^{1.5}$ edges where triangles are in an easily-computable $1-1$ correspondence with pairs $a,b$ satisfying ($\star$). We then run the algorithm for listing triangles. For each triangle in the list, we check if it corresponds to a solution for $\cxxx$.
This works because if the triangle-listing algorithm returns as many as $m$ triangles then, by above, at least one triangle corresponds to a correct solution. Hence, if listing can be done in time $m^{4/3-\e}$ then $\cxxx$ can be solved in time $(3n^{1.5})^{4/3-\e} = n^{2-\e'}$.

We now describe the graph. The graph is tripartite. One part has $\sqrt{n}
\times R$ nodes of the form $(b_h,x)$; another has $\sqrt{n} \times R$ nodes of
the form $(b_\ell,y)$; and the last part has $n$ nodes of the form $(a)$. Node
$(a)$ is connected to $(b_h,x)$ if $h(A[a+b_h]) = x$, and to $(b_\ell,y)$ if
$h(A[a+b_\ell]) = y)$.
Nodes $(b_h,x)$ and $(b_\ell,y)$ are connected if, letting $b = b_h + b_\ell$,
$h(A[b]) = x+y$.


We now count the number of edges of the graph. Edges of the form $(a)-(b_h,x)$ (or $(a)-(b_\ell,y)$) number $n^{1.5}$, since $a,b_h$ determine $x$. Edges $(b_h,x)-(b_\ell,y)$ number again $n^{1.5}$, since for each $b = b_h + b_\ell$ and $x$ there is exactly one $y$ yielding an edge.

The aforementioned 1-1 correspondence between solutions to $\cxxx$ and triangles is present by construction.
\end{proof}

\section{Reducing \fcliq~to \sixs}
\label{s-4-clique-6-sum}

In this section we prove the following connection between
solving \fcliq~and \sixs~over the group $\mathbb{Z}_3^t$.
 Although the next lemma is a simple
extension of Lemma \ref{lemma-first}, the fact that the
sum is over $\mathbb{Z}_3^t$ plays a crucial role in our
proof.  We do not see a simple way to carry through the
same reduction over $\mathbb{Z}$ or $\mathbb{Z}_2^t$.

\begin{restatable}{lemma}{4-clique}
\label{lemma-4-clique} Suppose that one can solve $\sixs$
on a set of $n$ elements over $\mathbb{Z}_3^t$ in time
$\tilde O(n^{1+\e})$ for $\e \geq 0$ and $t=O(\log n)$. Then,
given the adjacency list of a graph with $m$ edges,
$n=O(m)$ nodes,  one can decide if it contains a $\fcliq$
in time $\tilde O(m^{1+\e})$.
\end{restatable}
\begin{proof}
Similar to the proof of Lemma \ref{lemma-first}, consider
the family $S_i$ of $O(m)$ sets from Lemma
\ref{lemma-design}, with intersection size less than 1/11
of the set size. Assign to node $a$ the number $x_a$
whose decimal representation has 1 in the digits that
belong to $S_a$, and 0 otherwise. We look at $x_a$ as 
an element in $\mathbb{Z}_3^t$. For each edge $e=(a,b)$ 
let $Y_{(a,b)}:= (x_a + x_b)$. Return the output of 
$\sixs$ on the set $Y:=\{Y_{(a,b)}| (a,b) \in
E\}$. If there is a $\fcliq$ in the graph, there are 4
nodes, with an edge between any 2 of them, i.e.,${4
\choose 2}=6$ edges. The elements in $Y$ corresponding to
these 6 edges will sum to 0. This is because every node is
connected to 3 other nodes and the sum is over $\mathbb{Z}_3^t$.

On the other hand, if there are 6 elements $(x_a + x_{a'}),
(x_b + x_{b'}),(x_c + x_{c'}), (x_e + x_{e'}) ,(x_f + 
x_{f'})$ and $(x_g + x_{g'})$ in $Y$ that sum to $0$, 
then each element $x_i$ has to appear a  multiple of 
$3$ times. To see this, note that smaller than 1/11
intersection between any two subsets in $S_i$ assures
that no sum of less than 13 $x_i$ can sum to 0
unless each element appears a multiple of 3 times. If 
each $x_i$ appears exactly 3 times in the sum of 12 
elements, it means we have 4 nodes each one connected 
to the other 3 i.e., we have a $\fcliq$ in the graph. 
Notice if there is an element $x_a$ that appears 6 times, 
since the graph does not have self-loops or multiple 
edges, we can conclude that all the other elements 
are distinct and appear only once. We know that the 
sum of $6$ distinct elements over $\mathbb{Z}_3^t$ cannot 
be zero (due to the properties of $S_i$).
\end{proof}

\end{document}